\documentclass[a4paper, 11pt]{article}
\usepackage{amsmath, amsfonts, amsthm, mathrsfs, graphicx, cite, latexsym}

\title{Ecosystem Trust Profiles}
 
\author{Christoph F. Strnadl\footnote{Gaia-X European Association for Data and Cloud, Brussels, Belgium; \texttt{christoph.strnadl@gaia-x.eu}; ORCID: 0000-0003-4173-656X}}

\date{24 March 2026}

\begin{document}

\maketitle

\begin{abstract}

We define a method how digital ecosystems (including data spaces) may autonomously define and ``advertise'' credentials they issue or they trust in the form of so-called \emph{ecosystem trust profiles}. An ecosystem trust profile collects all (verifiable) credentials and issuers sorted by trust scope accepted (``trusted'') by a particular ecosystem.

We then show how a minimal trust relation between ecosystems may be defined using ecosystem trust frameworks of different ecosystems and explore a few of its properties.

A first application of the theory is given for a use case in the manufacturing realm where different international ecosystems need to agree on certain credentials for various scopes of trust such as identity, service compliance, and other conformance standards. 

We implement this requirement by identifying and discussing two different definitions of \emph{credential equivalence} for a given trust scope, one requiring additional cross-ecosystem governance or coordination, one not. The second approach demonstrates how to solve the so-called cross-ecosystem trust dilemma, that is, the problem how ecosystems can establish cross-ecosystem trust while, at the same time, allowing them to fully retain their sovereignty. A  \emph{fragility} theorem demonstrates that this sovereignty leads trust to be unstable without any additional coordination or governance mechanisms on top of (and outside to) ecosystem trust profiles.

We extend our method to data spaces in particular and propose a novel rigorous definition of cross-data space interoperability. 
This allows us to prove the proposition that the extent of interoperability between two data spaces is exactly determined by the amount of commonality in their respective ecosystem trust profiles.

\end{abstract}

\newline

\pagebreak

\tableofcontents

\section{Introduction}

\subsection{Trust}

Trust is a multi-faceted and ramified interdisciplinary topic hard to capture or define and even more difficult to formalize or automate \cite{Castelfranchi-2010, Paliszkiewicz-2022b}. Despite these methodological-conceptual challenges, the importance of trust and the principle way of its operations in an economic context were recognized as early as 1974 when Nobel Laureate Kenneth J. Arrow formulated \cite{Arrow-1974}:

\begin{quote} 
``Trust is an important lubricant of a social system. 
It is extremely efficient; it saves a lot of trouble to have a fair degree of reliance on other peoples’ word. […] Trust and similar values, loyalty or truth-telling […] increase the efficiency of the system, enable you to produce more goods […]. But they are not commodities for which trade on the open market is technically possible or even meaningful.''
\end{quote}

This overall observation was subsequently corroborated for e-commerce \cite{Gefen-2003}, digital business \cite{Paliszkiewicz-2022a}, and digital ecosystems \cite{Ion-2008, Cioroaica-2019, Cioroaica-2020}.

As we want to conceptualize (federated) trust in the specific setting of digital ecosystems, let us first further define what we understand as ``digital ecosystem'' before identifying the appropriate formalization of trust.

For our purposes, we combine the definitions of Adner \cite{Adner-2017} and Jacobides \cite{Jacobides-2018} and define an \emph{ecosystem} as follows (see also \cite{Strnadl-2024} for a rigorous and formal treatment):

\begin{quote}
An ecosystem is the ``not fully hierarchical alignment structure'' of a set
of partners (actors) ``with varying degrees of multilateral, non-generic complementarities\footnote{This property is needed to differentiate ecosystems for our purposes from two- or multi-sided platforms, where platform participants typically do not collaborate in any complementary way with their goods and services.} that need to interact in order for a focal value proposition
to materialize.''
\end{quote}

A \emph{data ecosystem} or \emph{digital ecosystem} can then be understood as the subset of all ecosystems where the focal value proposition is realized by sharing data or by employing digital or digitized services, respectively \cite{Strnadl-2023a}.

Our focus on ecosystem participants implies that (logical) formalizations of trust that try to capture the semantics of trust by employing modal (logic) operators such as $T(\phi)$, with the meaning of ``formula $\phi$ is trusted'' \cite{Tagliaferri-2019, Tagliaferri-2022}, or invoking belief functions $B(\phi)$ \cite{Ma-2006, Ma-2008}\footnote{with the intended meaning of ``trust is the outcome of observations leading to the belief that the actions of another may be relied upon''\cite{Ma-2008}} are insufficient. Instead, we will rely on the common and widely accepted \emph{relational} definition of trust or the closely related trustworthiness. 

In this approach, trust is defined as a relation between one party, typically called the \emph{trustor}, which relies on the beneficial (positive) behaviour of another party, the \emph{trustee}, in a situation which bears risk for the trustor. \cite{Huang-2006, Ferdous-2010, Cerutti-2013, Guven-2017, Carter-2023}.

\subsection{Cross-Ecosystem Trust Dilemma}

Ecosystems --- typically through an \emph{ecosystem governance authority} --- give themselves and their participants a set of policies and rules participants have to follow \cite{Gaia-X-2025a}. These rules and their automation in the form of suitable software components constitute an \emph{ecosystem's trust framework}. We need to stress that ecosystems, by design are completely autonomous in their choices regarding the trust framework and only depend on the set of actors and the particular ``alignment structure''.

Given this sovereignty over their own trust frameworks, the question then arises how ecosystems may trust each other when, in fact, no supra-ecosystem authority can force ecosystems to trust one another. We have called this the \emph{cross-ecosystem trust dilemma} \cite{Strnadl-2026b}:

\begin{quote}
\textbf{Cross-ecosystem trust dilemma}. How can digital ecosystems, which are fully autonomous in their choices regarding the definition and enforcement of their respective trust frameworks, establish mutual trust on behalf of their participants without any explicit bilateral or overarching governance (\textbf{independence condition}).
\end{quote}

The independence condition prohibits ``the existence of any supra-ecosystem authority able to somehow “force” or “orchestrate” ecosystems to trust one another.'' It also rules out ``the possibility that two or more ecosystems negotiate and explicitly (e.g., through contracts) agree on a common trust framework because such an approach will not scale'' \cite{Strnadl-2026b}.

This contribution solves the problem by introducing a (voluntary) disclosure mechanism based on a common minimal semantic. Otherwise, no further constraints are based on the interrelations of the ecosystems with each other. 

However, this lack of additional coordination or governance mechanisms between (or on top) of a set of interacting ecosystems severely restrict the properties of any suitable \emph{trust relation} between them. As we will rigorously show in section \ref{sec-fragility-theorem}, this implies that, in general, trust relations between ecosystems are not stable over time.

\section{Ecosystem Trust Profiles}			\label{sec-theory}

%\pagevalues

\newtheorem{definition}{Definition}
\newtheorem{proposition}{Proposition}	
\newtheorem{lemma}{Lemma}
\newtheorem{corollary}{Corollary}

\newcommand{\mE}{\mathscr{E}}
\newcommand{\mS}{\mathscr{S}}
\newcommand{\mP}{\mathscr{P}}
\newcommand{\mC}{\mathscr{C}}
\newcommand{\mT}{\mathscr{T}}
\newcommand{\mO}{\mathscr{O}}					%% realm of organisations
\newcommand{\mD}{\mathscr{D}}					%% realm of data spaces

The trust framework of an ecosystem, \emph{inter alia}, specifies so-called \emph{trust service providers} which issue (verifiable) credentials attesting certain properties for ecosystem-relevant entities. Typical examples are participant credentials for attesting membership of an entity in the particular ecosystem, or service credentials attesting the compliance of a service offered by one of the ecosystem members with some pre-defined criteria (e.g., a \emph{compliance document} such as the Gaia-X Compliance Document \cite{Gaia-X-2025b}.

We formalize this as follows. Let $\mP$ be a set of trust service providers $p \in \mP$ issuing (verifiable) credentials $c \in \mC$. We further assume that credentials are sorted by a (trust) scope $s \in \mS$. From a model-theoretic perspective, having trust \emph{scopes} is not necessary at all. They serve as an aggregation or characterization of individual credentials. However, in section \ref{sec-application} we will demonstrate how this model element may be used to construct an elegant formulation of \emph{credential equivalence}.

Simplifying the construct of a ``trust statement'' \cite{Guven-2017}, we then formalize those portions of an ecosystem's trust framework relevant for our purposes as a set of so-called \emph{trust propositions}.  

In \cite{Guven-2017}, a trust statement is a 4-tuple $\langle T, \wp, c, \delta \rangle$ with $T$ the trustee, $\wp$ a predicate, $c$ the context (e.g., conditions, assumptions) within which the trust statement is valid, and $\delta$ the time period of its validity. Comparing this with our informal example of above, we find that we have to equate $T$ with our trust service provider(s), $\wp$ with the set of (verifiable) credentials issued by the particular trust service providers, and $c$ with the trust scope. At this point of the formalization, we will not need a validity time $\delta$.

This allows us to write the trust proposition "Trust service provider $p$ issues a credential $c$ for scope $s$" in the form of the triple $t = (s,p,c) \in \mT \equiv \mS \times \mP \times \mC$. An \emph{Ecosystem Trust Profile} corresponds then to the list of all trust propositions accepted by a particular ecosystem, that is, where the respective credentials issued by the identified issuers are accepted by all ecosystem participants.

\begin{definition}[Ecosystem trust profile] 	\label{def-eco-trust-profile}
Let $P \subseteq \mP$ be a set of trust service providers issuing credentials $c \in \mC$ for a scope $s \in \mS$ and let $T \subseteq \mT = \mS \times \mP \times \mC$ be a set of trust propositions. Then the structure $E = \langle P, T \rangle$ is called an \emph{ecosystem trust profile}.

We further fix the following abbreviations:
\begin{enumerate}

\item The set $S = S(E) = \{ s \; | \; \exists \, (s,p,c) \in T \}$ is called the \emph{(trust) scopes} accepted by $E$.
\item The set $C = C(E) = \{ c \; | \; \exists \, (s,p,c) \in T \}$ is called credentials accepted by $E$. Because $S$ and $C$ are uniquely defined for a fixed $E$, we will write $s \in E$ or $c \in E$ when, in fact, $s \in S(E)$ or $c \in C(E)$.
\item The set $P$ is called the set of \emph{domestic} trust service providers of $E$.
\item All non-domestic trust service providers $p' \in \mP \setminus P$ are called \emph{foreign} trust service providers for $E$.

\end{enumerate}

In order to simplify our language we will often equate the term "\emph{ecosystem} $E$" with "an ecosystem characterized by  \emph{ecosystem trust profile} $E$".
\end{definition}

The intended meaning of $P$ is to denote those trust service providers which are under the governance of a particular ecosystem or otherwise specifically related to an ecosystem. In that sense, a provider $p \in P$ may be thought of issuing credentials on behalf of the ecosystem $E$. Ecosystem participants (which are not modelled directly in this approach) are always assumed to trust credentials issued by domestic trust service providers\footnote{It is actually this small \emph{observation} which allows us to have a flat list of trust propositions.}. 

All other trust service providers $p' \in \mP$ with $p' \notin P$ are called \emph{foreign} trust service providers (\emph{scil.} for ecosystem $E$). A triple $(s,p',c)$ may then be understood as the assertion that ecosystem $E$ accepts credentials $c$ for trust scope $s$ issued by a foreign trust service provider $p'$ \emph{not belonging} to or \emph{under the control} of $E$. Often, the \emph{foreign} trust service provider $p'$ belongs to a different ecosystem $E'$, but that does not always need to be the case and $p'$ may not directly belong to any ecosystem of our universe.

We will now use the concept of foreign trust service providers to define the fundamental trust relationship between ecosystems.

\begin{definition}[Foreign ecosystem trust relation]	\label{def-eco-trust-relation}
Let $\mE$ be a set of ecosystem trust profiles and let $E_i = \langle P_i, T_i \rangle \in \mE$ 
and $E_j = \langle P_j, T_j \rangle \in \mE$. Then the (ecosystem trust) relation "ecosystem $E_i$ trusts a (foreign) ecosystem $E_j$ regarding trust scope $s$", written as $E_i \sqsubset_s E_j \in \mE \times \mE$, is defined as follows:
\begin{equation}						\label{eq-eco-trust-relation}
	E_i \sqsubset_s E_j \; \Longleftrightarrow \; \exists \, c \ 
	         \exists \, p_j  \in P_j \setminus P_i: 
		(s,p_j,c) \in T_i \cap T_j.
\end{equation}
\end{definition}

The key point here is that ecosystem $E_i$ trusts a \emph{foreign} trust service provider $p_j$ of ecosystem $E_j$ which issues a credential $c$.\newline

\textbf{Example.} We use the \emph{trust scope} to model constructs where ecosystems would agree on a kind of weak common semantics of the credentials like for a scope ``identity'' or ``participant'' property. The Canadian ecosystem would have a Canadian company number credential for scope ``Identity'' and Factory-X would have a Catena-X credential for the same scope — and they would mutually recognize the other ecosystem's credential for that scope.

\begin{lemma}[Irreflectivity]
The relation $\sqsubset_s$ is irreflexive, that is, $E \not\sqsubset_s E$ for all $E = \langle P, T \rangle \in \mE$ for all trust scopes $s \in S$.
\end{lemma}
\begin{proof}
For symmetric arguments to $\sqsubset_s$, the right-hand side of eqn.~(\ref{eq-eco-trust-relation}) requires to find a "foreign" $p \in P \setminus P$.
As this is impossible because we only have \emph{domestic} trust service providers, $P \setminus P = \emptyset$, and
the relation fails\footnote{Kudos to Christopher Grillo from Cybernetica for this.}. This is also the reason for calling the relation ``foreign'' ecosystem trust relation.
\end{proof}

Let's look a little bit deeper in the set of \emph{trust propositions} $T_i$ and $T_j$ for two ecosystems $E_i = \langle P_i, T_i \rangle$ and $E_j = \langle P_j, T_j \rangle$. Assume that the intersection of the sets of trust propositions of both ecosystems is non-empty, i.e., $T_i \cap T_j = \{t,... \}  \neq \emptyset$ for at least one trust proposition $t = (s,p,c) \in T_i \cap T_j$. Then we can discriminate four cases for the trust service provider $p$\footnote{We use the logical signs $\top$ and $\bot$ to denote logical truth and falsehood. We also assume, without any loss of generality, that all providers $p$ belong to at least one ecosystem $E \in \mE$}.

\begin{center}
\begin{tabular}{cccl}
\hline\hline
 & $p \in P_i$ & $p \in P_j$ & Consequences \\	
 \hline
 1 & $\bot$ & $\bot$ & 
 			$\exists E^\star \in \mE$ with $E_i \sqsubset_s E^\star$ and $E_j \sqsubset_s E^\star$. \\
 2 & $\bot$ & $\top$ & $E_i \sqsubset_s E_j$. \\ 
 3 & $\top$ & $\bot$ & $E_j \sqsubset_s E_i$. \\
 4 & $\top$ & $\top$ & $E_i \sqsubset_s E_j$ and $E_j \sqsubset_s E_i$ \\
\hline\hline
\end{tabular}
\end{center}

Obviously, only cases (1) and (4) are revealing.
\begin{itemize}
\item Case (1) indicates the situation where two ecosystems both depend on the same certificate issues by provider in another ecosystem. This is a situation, for instance, for identity credentials where ecosystems typically depend on another ``ecosystem'' of suitable identity providers or issuers, e.g., eIDAS.
\item Case (4) is an even stronger than the two symmetric trust relations seem to indicate because both ecosystems actually share a common trust service providers and both accept the same credential (for the same scope). The situation of ``mutual trust'' encoded in $E_i \sqsubset_s E_j \land E_j \sqsubset_s E_i$ is weaker as it may already hold in case two quite distinct trust propositions $(s,p_j,c_1) \in T_i$ and $(s,p_i,c_2)\in T_j$ hold with $p_j \in P_j$ and $p_i \in P_i$.
\end{itemize}

We can use case (4) to define a much stronger trust relationship \emph{direct mutual trust} as follows.

\begin{definition}[Direct mutual trust]
Let $E_i = \langle P_i, T_i \rangle$ and $E_j = \langle P_j, T_j \rangle$ be two ecosystem. Then the relation $\Box \subseteq \mE \times \mE$, \emph{direct mutual trust} between ecosystems $E_i$ and $E_j$, $E_i \; \Box \; E_j$, is defined as follows.
\begin{equation}
	E_i \; \Box \; E_j \Leftrightarrow T_i \cap T_j \neq \emptyset.
\end{equation}
\end{definition}

Note that $\Box$ still is not transitive in general because $E_i \; \Box \; E_j$ may be true because of a trust provider $p$ with trust proposition $(s,p,c)$ and $E_j \; \Box \; E_k$ may hold because of another trust provider $p^\star$ for trust proposition $(s^\star,p^\star, c^\star)$.

We use the characterization of trust propositions by a kind $s$ to introduce the notion of a \emph{trust realms} as follows.

\begin{definition}[Trust realm]			\label{def-trust-realm}
Let $\mE$ be a set of ecosystem trust profiles and write $E = \langle P, T \rangle \in \mE$ for an arbitrary element of this set. Then for a $s \in S(E)$ the (domain of the function) $\rho = \rho(s)$ is called the \emph{trust realm of scope $s$} and defined as follows:
\begin{align*}
	\rho : S(E) & \to 		\wp(\mE)	\\
	          s & \mapsto	\rho(s) = \{ E \in \mE \ | \ \exists \, (s,p,c) \in T \}.
\end{align*}
\end{definition}

The trust realm simply collects all ecosystems accepting at least a single credential of scope $s$ issued by one "trusted" trust service provider.

\section{Application} \label{sec-application}

\subsection{Manufacturing Use Case}	\label{sec-mfg-use-case}

The following use case drawn from the international Manufacturing-X (IMX) project demonstrates how the theoretical model above may be applied to concrete instantiations of cross-ecosystem trust relationships.

Assume a set of independent (manufacturing) ecosystems which all individually recognize different credentials attesting membership of companies in the respective ecosystem, adherence of services to certain standards, or conformity of endpoints to some pre-agreed set of criteria. Due to the globalization of supply chains, an enterprise will have to participate in two or more such ecosystems at the same time. In the same vein, services offered and certified within one ecosystem may very well be needed by other ecosystems as well (just think of services related to digital product passports).

In order to avoid the overhead of requiring ecosystem participants having to re-certify their identity, service compliance, or other conformance credential for essentially the same purpose or scope, the idea of this ``ecosystem'' of ecosystems is to allow and define a kind of mutual recognition of credentials for the various scopes.

The user story is captured in figure \ref{fig-catalog} of a \emph{catalog} listing which (verifiable) credentials (VC) issued by which trust service provider (TSP) of an ecosystem (identified by the \texttt{Ecosystem ID}) shall be deemed equivalent for a given trust scope, for instance, for identity (\texttt{imxc:Identity}).

\begin{figure}[htbp]
\begin{center}
\includegraphics[width=\textwidth]{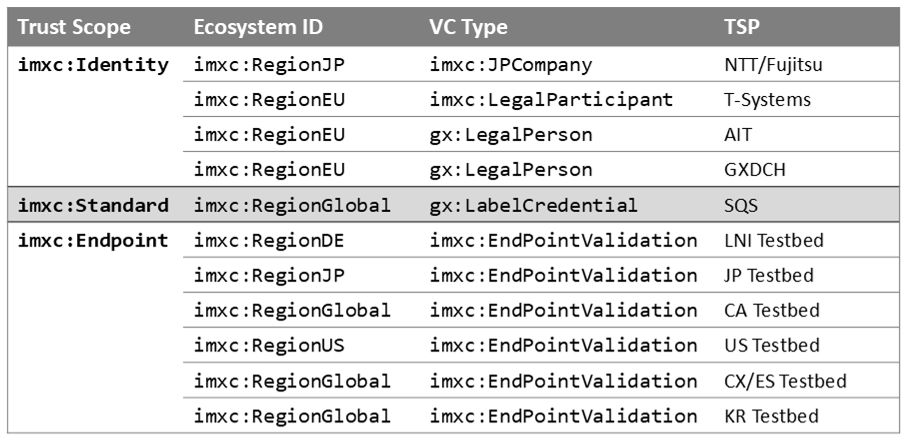}
\end{center}
\caption{Catalog of ecosystems, trust service providers, and credentials (VC - verifiable credential; TSP - trust service provider)}
\label{fig-catalog}
\end{figure}

From the previous exposition of our model of ecosystem trust profiles, it is evident that the formulation of such a \emph{credential equivalence} cannot be achieved at the current level of formalization.

The necessary extension of our theoretical model to achieve this requirement is given in the next section \ref{sec-theoretical-extension}.

\subsection{Theoretical Extensions} 	\label{sec-theoretical-extension}

The following three subsections provide different approaches of realizing the requirements stated in section \ref{sec-mfg-use-case} above.

\subsubsection{Top Down - Brute Force}	\label{sec-top-down-brute-force}

We first define the set of credentials $c$ accepted within a particular trust scope $s$ as follows\footnote{We assume the set of ecosystems $\mE$ to be suitably fixed for the following.}:

\begin{definition}[Credentials for a given trust scope]
Let $E_i = \langle P_i, T_i \rangle \in \mE$ be a collection of ecosystems. Then we define the set of all credentials issued and accepted under trust scope $s$, written $C(s)$, as follows:
\begin{equation}
	C(s^\star) = \{ \, c \; | \, \exists \, E_i = \langle P_i, T_i \rangle \in \mE : 
			(s^\star, p, c) \in E_i \, \}.
\end{equation}
\end{definition}

Then the following proposition on the equivalence of credentials formulates sufficient conditions under which the construction of section \ref{sec-mfg-use-case} may hold.

\begin{proposition}[Credential equivalence - Version 1]	\label{prop-cred-equiv-1}
Let $\mC_0$ denote the set of all credentials used in $\mE$\footnote{We cannot really use the previous set, $\mC$, because there might not be trust propositions for every credential $c \in \mC$ in our universe.},
\begin{equation}
	\mC_0 = \bigcup_{s \in \mS} C(s),
\end{equation} 
and assume the $C(s)$ to be a \emph{partition} of $\mC_0$. Then the relation $\cong \; \subseteq \mC_0 \times \mC_0$, defined by
\begin{equation}			\label{equ-credential-equiv-1}
	c_1 \cong c_2  \quad \Leftrightarrow \quad \exists 
		\;  s \in \mS : \; c_1 \in C(s) \land c_2 \in C(s)
\end{equation}
is an equivalence relation.

Informally, $c_1 \cong c_2$ denotes that credentials $c_1$ and $c_2$ are regarded as equivalent for proving the proposition underlying their (shared) trust scope.

\end{proposition}

\begin{proof}
While the proof, technically speaking, is trivial as every partition automatically defines a canonical equivalence relation, we will repeat the proof for informational purposes.

Recall that an equivalence relation has to fulfil the laws of (i) reflexivity ($c \cong c$), (ii) symmetry ($c_1 \cong c_2 \Rightarrow c_2 \cong c_1$), and (iii) transitivity 
( $c_1 \cong c_2  \, \land \, c_2 \cong c_3 \; \Rightarrow \; c_1 \cong c_3$).

Also remember that a partition of the set $\mC_0$ into subsets $C(s)$ (indexed by $s$) fulfils the following conditions:
\begin{eqnarray}
	\bigcup_{s \ \in \ \mS} C(s) & = & \mC_0,			\label{equ-partition-1}		\\
	C(s) \cap C(s^\star) & = & \emptyset \quad \text{for  } s \ne s^\star. 
													\label{equ-partition-2}
\end{eqnarray}

Our proof will then simply proceed along the three conditions (i) - (iii).

\begin{enumerate}
\item \textbf{Reflexivity}. Because of eqn.~(\ref{equ-partition-1}), $\exists \, s \in \mS:
c \in C(s)$. The rest follows immediately from the symmetry of the antecedent of (\ref{equ-credential-equiv-1}) in its arguments.

\item \textbf{Symmetry}. Assume $c_1 \cong c_2$. That means that $\exists \, s \in \mS:
c_1 \in C(s) \land c_2 \in C(s)$. Because of the commutativity of $\land$, this can be immediately rewritten as $c_2 \in C(s) \land c_1 \in C(s)$ which directly implies
$c_2 \cong c_1$ as per the definition.

\item \textbf{Transitivity}. Assume $c_1 \cong c_2$ and $c_2 \cong c_3$. That means that
\begin{eqnarray}
	\exists \; s \in \mS & : & c_1 \in C(s) \land c_2 \in C(s),	\label{equ-trans-1}) \\
	\exists \; s^\star \in \mS & : & c_2 \in C(s^\star) \land c_3 \in C(s^\star).												\label{equ-trans-2}
\end{eqnarray}
Because of (\ref{equ-partition-2}), $s = s^\star$. That means that the second assertion of (\ref{equ-trans-2}) may be rewritten as $c_3 \in C(s)$ for the same $s$ as in 
equ.~(\ref{equ-trans-1}). Combining this with the first assertion of (\ref{equ-trans-1}), we find that
\begin{equation}
	\exists \; s \in \mS : c_1 \in C(s) \land c_3 \in C(s),
\end{equation}
which is the definition of $c_1 \cong c_3$.
\end{enumerate}
\end{proof}
This construction essentially says that every credential issued within a given trust scope is equivalent --- which is exactly the meaning purported by the catalogue of figure \ref{fig-catalog} above.

\textbf{Problem.} The problem with this approach is that it opens up the credential
equivalence relation to outside Byzantine (i.e., malevolent)``imposters'' who are able 
to insert their credentials into the trust scope easily by \emph{simply declaring}
malicious credentials. Assume $E^\star = \langle P^\star, T^\star \rangle$ to be
a malicious ecosystem trying to insert a new malicious certificate $c^\star$ into
trust scope $s$. $E^\star$ can simply do that by issuing $c^\star$ through one of its providers $p^\star \in P^\star$, thus asserting $(s, p^\star, c^\star) \in T^\star$. Then, by (\ref{equ-credential-equiv-1}), $c^\star \cong c \;\,  \forall \, c \in C(s)$.

Avoding this relies either (i) on the cooperation of the different ecosystems to only issue trust propositions in line with the other ecosystems, or (ii) on the maintainer of the meta-registry holding all the trust proposition advertisements. Both preconditions may be (more or less easily) met by using a suitable decentralized consensus mechanisms of the sort amply available by today's distributed ledger technologies (DLT).

\subsubsection{Cautious Approach - Minimal consensus} \label{sec-cautious-approach}

One way of avoiding the problem of ``imposters'' of the credential equivalence relation proposed in section \ref{sec-top-down-brute-force} is to relay on the construction used at the end of section \ref{sec-theory} when we introduced the notion of a \emph{direct trust relation} between ecosystems.

\begin{definition}[Credential equivalence - Version 2]	\label{def-cred-equiv-2}
Let $\mE = \{ E_1, E_2,...,E_N \}$ be a set of ecosystems $E_i = \langle P_i, T_i \rangle$ with credentials $\mC$. Define $T_0$ to be the set of trust propositions accepted by \emph{every} ecosystem in $\mE$, that is 
\begin{equation}			\label{equ-credential-T0}
	T_0 = \bigcap_{i = 1}^N \, T_i.
\end{equation} 
Then the relation $\cong_s \; \subseteq \mC \times \mC$
\begin{equation}			\label{equ-credential-equiv-2}
	c_1 \cong_s c_2  \quad \Leftrightarrow \quad 
			\exists \, (s, p, c_1) \in T_0 \land \exists \, (s, p', c_2) \in T_0	
\end{equation}
has the (informal) meaning of credential $c_1$ being equivalent to credential $c_2$, that is: $c_1 \cong_s c_2$. Note that this equivalence does not extend equivalence beyond the borders of trust scopes.
\end{definition}

Again, this definition has a nice property as the following proposition shows.

\begin{proposition}(Credential equivalence Version 2 - equivalence relation) \label{prop-cred-equiv-2}
The relation $\cong$ defined in definition \ref{def-cred-equiv-2} is an \emph{equivalence relation}.
\end{proposition}
\begin{proof}
Reflexivity and symmetry are trivially satisfied. For transitivity, assume $c_1 \cong c_2$ and $c_2 \cong c_3$. This means,
\begin{equation}			\label{equ-prop-2-1}
	c_1 \cong c_2 \; \Rightarrow \; \left\{
		\begin{array}{l}
			\exists \, (s, p_1, c_1) \in T_0, 	\\
			\exists \, (s, p_2, c_2) \in T_0,
		\end{array}
		\right.
\end{equation}
and
\begin{equation}			\label{equ-prop-2-2}
	c_1 \cong c_2 \; \Rightarrow \; \left\{
		\begin{array}{l}
			\exists \, (s, p'_2, c_2) \in T_0, 	\\
			\exists \, (s, p_3, c_3) \in T_0.
		\end{array}
		\right.
\end{equation}
Then take the first equation of (\ref{equ-prop-2-1}) and the second equation of (\ref{equ-prop-2-2}). This immediately yields
\begin{equation*}
	(s, p_1, c_1) \in T_0 \land (s, p_3, c_3) \in T_0 \; \Rightarrow \; c1 \cong c3.
\end{equation*}
\end{proof}

This approach clearly eliminates the problem of ``imposters'' present in proposition \ref{prop-cred-equiv-1} by virtue of equation (\ref{equ-credential-T0}). The construction is also stable against ``withdrawal'' attacks where an ecosystem $E'$ may decide to remove a credential $c'$ it issues from the common pool $T_0$ because all other ecosystems $E \in \mE \setminus \{ E' \}$ may freely continue to issue and trust this credential $c'$. The only remaining problem is if this ecosystem is a \emph{monopoly} provider of $c'$.

\subsubsection{A Fragility Theorem on Trust}		\label{sec-fragility-theorem}

It is evident from the discussions in the preceding sections that \emph{any} trust relation between credentials, and, \emph{a forteriori}, between ecosystems, is difficult to define and/or to maintain over time. This is nicely captured in the (otherwise not so satisfying) proposition \ref{prop-fragility}:

\begin{proposition}[Fragility of ecosystem trust] \label{prop-fragility}
Let $\mE = \{ E_1, E_2,...,E_N \}$ be a (non-empty) set of ecosystems (ecosystem trust profiles) of the form
$E_i = \langle P_i, T_i \rangle$, and let the relation $\preceq \; \subseteq \mE \times \mE$ be an \emph{arbitrary} trust relation between ecosystems with the (informal) meaning of $E_i \preceq E_j$ as of ``$E_i$ trusts $E_j$''. 

Further, let the map $\kappa: \mE \times \mE \to \wp(\mT)$\footnote{$\wp(A)$ denotes the power set of $A$.} 
denote the necessary conditions for ecosystem $E_i$ in terms of which trust propositions $t \in \mT$ are required to render $E_i \preceq E_j$ true. Formally,
\begin{equation}			\label{equ-kappa-impl}
	E_i \preceq E_j \; \Rightarrow \; \kappa(E_i, E_j) \ne \emptyset \; 
										\subseteq \, T_i \cap T_j.
\end{equation}
We require 
$\kappa(E_i, E_j) \ne \emptyset$ to avoid situations where the trust relation depends on a state of affairs not captured in the trust propositions shared by \emph{both} ecosystems.\footnote{If $\kappa(E_i,E_j) \not\subseteq T_j$, the trust of ecosystem $E_i$ in $E_j$ would not depend on any trust proposition of $E_j$. We will not further elaborate on this fairly pathological situation where one's trust in another does not depend on the other at all.}

Then this arrangement is not \emph{stable}\footnote{We admit that the notion of ``stability'' is not formalized yet. Nevertheless, we hope the reader already appreciates the consequences of this informal ``proof''.} in the case of sovereign ecosystems without any further coordination/governance mechanism.
\end{proposition}
\begin{proof}
Assume, without loss of generality, that $E_1 \preceq E_2$, that is, $E_1$ as trustor trusts another ecosystem $E_2$, the trustee. 
Then there exists at least one $t \in \kappa(E_1, E_2)$ with $t \in T_2$, that is, depending on the trustee (ecosystem).

The trustee $E_2$ then ``decides'' to withdraw the assertion of trust proposition $t$. $E_2$ is able to do that because it is, according to the assumptions, (i) sovereign in its decisions, and (ii) not bound by any external coordination or governance requirements forbidding such a decision.

This means, that now $\kappa(E_1,E_2) \not\in T_2$ because $E_2$ has retracted
$t$ from $T_2$. By contraposition of eqn.~(\ref{equ-kappa-impl}), we get $E_1 \not\preceq E_2$. This means, that after such a change by the trustee ($E_2$), the original trust relation no longer holds: $E_1$ no longer trusts $E_2$ as a consequence of this move.
\end{proof}

Note that we do \textbf{not} claim that there cannot be collections of ecosystems with a useful notion of a trust relation. What our \emph{fragility theorem} proves is that this is an unstable state of affairs in the absence of additional ``stabilization'' mechanisms. Any such mechanism---as shown in the proof---needs to curtail an ecosystem's ability to arbitrarily withdraw (!) certain trust propositions from its (advertised) set.

Also note that definition \ref{def-cred-equiv-2} and proposition \ref{prop-cred-equiv-2} on our second version of credential equivalence do \textbf{not} violate proposition \ref{prop-fragility} above because definition \ref{def-cred-equiv-2} does \emph{not} define a trust relation on \emph{ecosystems} but on \emph{credentials}.

\section{Implementation}

\subsection{Approach}

The concepts and methods laid out in the previous sections have been implemented by the Gaia-X Association (https://gaia-x.eu/) in the form of a minimum viable product (MVP) at this point in time using the following two elements:

\begin{enumerate}

\item \label{item-ontology}
\textbf{Ecosystem Trust Profile Ontology}: A machine-readable schema\footnote{In fact, it is a full-fledged OWL ontology specified in LinkML \cite{LinkML-2024}.} for specifying ecosystem trust profiles in an automatable way	
	
\item \label{item-GXMR}
\textbf{Gaia-X Meta-Registry} (GXMR): A mechanism for end users and machines alike to discover and retrieve ecosystem trust profiles for various ecosystems

\end{enumerate}
Strictly speaking, the \ref{item-GXMR}nd component is not needed but only a mechanism to ``somehow'' (a) find and (b) read the ecosystem trust profile of a particular ecosystem, i.e., the ecosystem with which a participant of another ecosystem wants to do transactions with). 

A typical implementation pattern for requirement (a) is, of course, a suitable form of a \textbf{catalog} holding and presenting the ecosystem trust profiles of (several) ecosystems. Catalogs holding this type of trust-related meta-data are often called (trust) registries. We are not aware of any real-world conceptualization or implementation of any such \textit{Meta-Registry}.

Requirement (b) may be implemented by a suitable protocol between a system querying an ecosystem for its trust profile and an endpoint on another system providing this information. This has indeed been proposed in the form of the Trust Reqistry Query Protocol (TRQP) \cite{ToIP-2025a}.\footnote{The Trust over IP working group did also start work on defining a generic ontology for ecosystem registry entries, the Trust Register Query Language (TRQL) \cite{ToIP-2025b}. The specification remains at a decidedly generic level, though, and cannot be directly used for our purposes.}

\subsection{Ontology}

In our ontology, an Ecosystem is anything with an \texttt{ecoID} (the class name in the ontology in figure \ref{fig-ontology}) and a (non-null) set of \texttt{ecoTrustScope}s with one or more trust service providers (TSPs) (\texttt{ecoTSP}) issuing one or more types of verifiable credentials (VCs) (\texttt{ecoCredentialType}). The outline of the ontology is specified in figure \ref{fig-ontology} in the form of a class diagram.

\begin{figure}[htbp]
\begin{center}
\includegraphics[height=8cm]{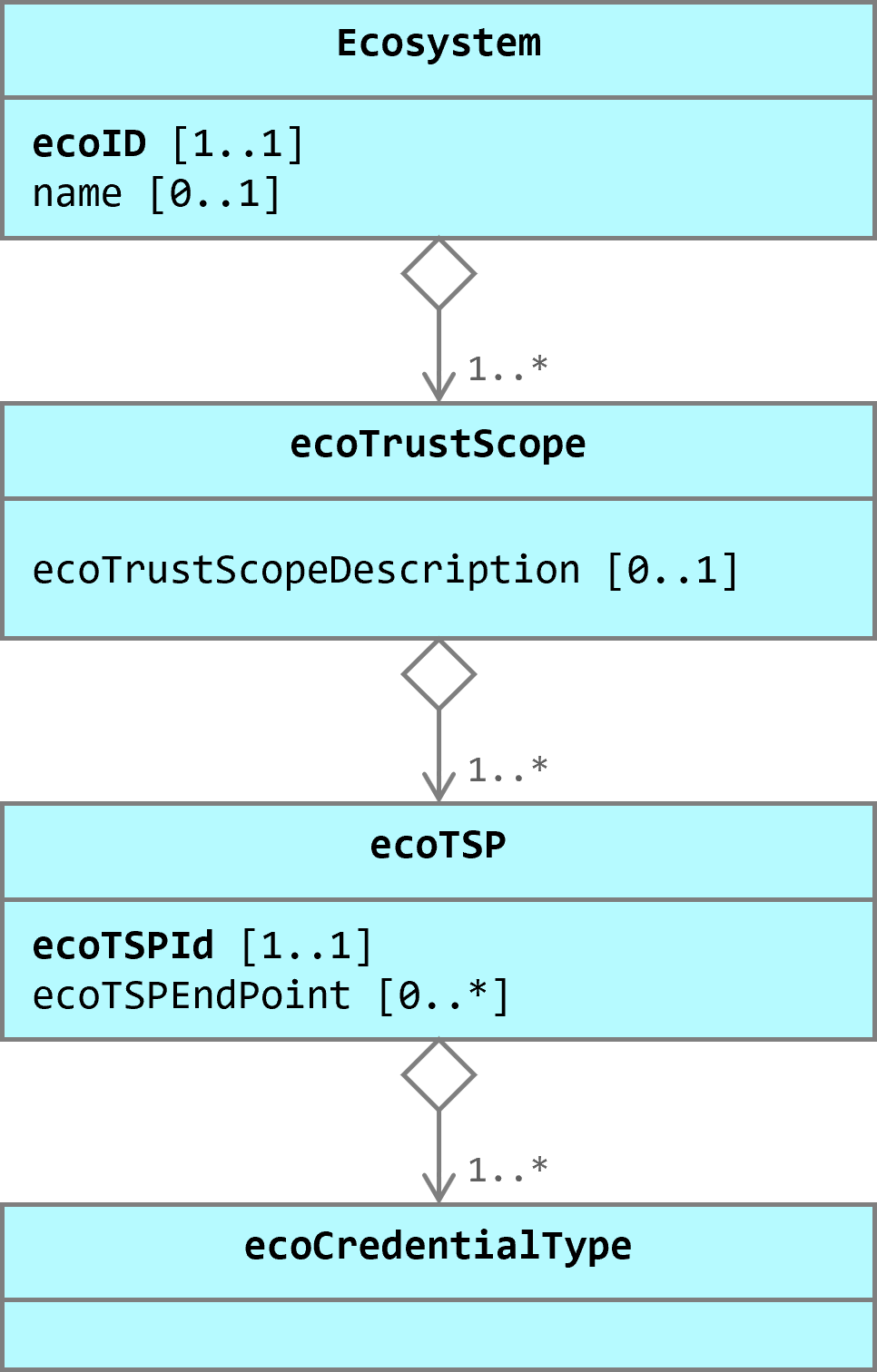}
\end{center}
\caption{Class diagram of the ecosystem trust ontology}
\label{fig-ontology}
\end{figure}

\textbf{Notes}
\begin{itemize}

\item We do not require that an ecosystem recognizes ``participation'' or ``identity'' as one of its core ``properties''. Any identifiable collection of trust scopes, trust services providers, and VCs will do.
\item Many ecosystems will indeed issue ``Participant Credentials'' like Catena-X or Gaia-X itself which is perfectly possible in our scheme.
\item Nevertheless, our approach allows ecosystems of IoT device credentials or Digital Twins or DPPs – as long as one or more TSPs issue VCs to that effect.
\item Currently, semantics of the \texttt{ecoTSPEndpoint} property are not defined or constrained at all. In view of the many different VCs and implementations, any such standardization effort may also be futile.
\item \texttt{ecoTrustScope} identifiers are currently not constrained in any way – as long as they are unique within a single \texttt{ecoSystem}. Maybe we see some standardization in the future here like for trusted registries, trusted identity providers, and other.

\end{itemize}

\textbf{Location}:
You will find the complete (machine-readable) ontology at https://gaia-x-meta-registry-99af54.gitlab.io/ecosystems/gaia-x/

\subsection{Gaia-X Meta-Registry}

In parallel, the CTO team of the Gaia-X AISBL is working on providing a suitable Gaia-X Meta-Registry mechanism for collecting and exposing ecosystem trust profiles. The current state of this work is published as a MVP at https://gaia-x-meta-registry-99af54.gitlab.io/ecosystem\_ontology\

\section{Data Spaces}

\subsection{Data Space Trust Frameworks}

This section expands the formalization towards full-fledged data spaces, that is, ecosystems with trust profiles specifically allowing the mutual sharing of data.
For that, we will need to introduce the following new concepts to augment the ecosystem trust profile defined in section \ref{sec-theory} above towards the definition of a full-fledged \textbf{trust framework} for a data space.

\begin{definition}[Trust framework]			\label{def-trust-framework}
Let $E = \langle P, T \rangle \subseteq \mE$ be an ecosystem trust profile. Assume that all rules and policies (i.e., sets of rules) applicable for all participants of this data space to be captured in the set $R = \{r_1, r_2,... \}$. 

Let the binary relation $\vdash \; \subseteq T \times R$, in particular
$t \vdash r \equiv (t,r) \in \; \vdash$, denote the fact that trust proposition $t \in T$ \emph{asserts} or \emph{attests} rule $r$. For consistency reasons we require that 
\begin{equation}
	\forall \,r \in R: \exists \, t \in T: t \vdash \, r .			\label{eq-a-consistency}
\end{equation}
 For ease of reference, we sometimes set $A \equiv \; \vdash$, the set of \emph{assertions} describing claims and evidences $t$ fulfilling some rules $r$.

Then the structure $F = \langle T, R, \, \vdash \, \rangle = \langle T, R, A \rangle$ is called the \emph{trust framework} of ecosystem $E$.
\end{definition}

The consistency relation (\ref{eq-a-consistency}) is needed to ensure that the trust framework does not contain rules which cannot be satisfied by any permissible trust proposition $t \in T$. \\

\textbf{Example.} A very small example of a trust framework is given in the following:
\begin{eqnarray}
	T & = & \{ (\iota, c_I, t_0), (\mu, c_M, t_0) \},		\label{eq-example-T}\\
	R & = & \{ r_I, r_M \},									\label{eq-example-R}\\
	A = \; \vdash & = & \{ (\iota, c_I, t_0) \, \vdash \, r_I, \; (\mu, c_M, t_0) \, \vdash \, r_M \}.	\label{eq-example-A}
\end{eqnarray}
Eqn.~(\ref{eq-example-T}) lists two trust propositions, one for trust scope $\iota$ (identity), and another one for membership trust scope $\mu$, both of which being attested by the same trust service provider $t_0$ using two different credentials, $c_I$ for identity, and $c_M$ for membership.

The data space recognizes two different rules for identity and membership, $r_I$ and $r_M$, in eqn.~(\ref{eq-example-R}), and accepts the two credentials mentioned above issued by TSP $t_0$ to attest these rules as per eqn.~(\ref{eq-example-A}).\\

The final missing ingredient for working with data space trust frameworks are the participants of a data space which we take from the 
realm $\mO = \{O_1, O_2,...\}$ of organisations potentially participating or supporting data spaces. Then we define the trust framework-related parts of a \textbf{data space} as follows.

\begin{definition}[Data space]			\label{def-data-space}
Let $E = \langle P, T \rangle \subseteq \mE$ be an ecosystem trust profile, $F = \langle T, R, \vdash \, \rangle$ a trust framework building on $E$, and
$O \in \mO$ a set of data space participants.

Then the structure $D = \langle O, T, R, \vdash \, \rangle$ is called a \emph{data space}.
\end{definition}

\begin{quote}
\textbf{Note.} We readily acknowledge that the current formalization of a \emph{data space} is heavily tilted towards trust-related concepts and does not contain other highly important elements of any data space, for instance, the data to be shared, or the components for realizing trusted data transactions  such as data space participant agents or connectors.
\end{quote}

Data spaces are for sharing data, which we formalize as follows:

\begin{definition}[Data sharing relation]			\label{def-data-sharing-rel}
The relation $\rhd \subseteq \mO \times \mO$ is called a data sharing relation with the intended semantics that $p \rhd q$ signifies that a data provider or
data holder $p$ is sharing data or making data accessible to a data consumer or data user $q$.
\end{definition}

We will now use this formal model to derive necessary general conditions for data transactions taking place.

\subsection{Necessary general conditions for data transactions} 

For the following, let $D = \langle O, T, R, \vdash \, \rangle$ denote a particular data space.
Data spaces unite a set of participants who all agree to the same set of governance rules, in particular, to the same data space trust framework, in order to facilitate and allow fast scaling of data exchange amongst all participants. Assume, a data provider $p \in O$ wants to share data with data user $q \in O$, that is, both parties to this data transaction want to allow $p \rhd q$.

Now, the core function of a data space trust framework is to specify all general conditions, criteria, rules, or policies, which need to be fulfilled in order to make this happen. We denote with $R_\rhd \subseteq R$ all rules of the data space trust framework related to allow a data transaction. Let us turn to the two partners of this data transaction:

\begin{enumerate}
\item \textbf{Data provider $p$ perspective.}
For a data transaction to take place, a provider $p$ will require the data consumer $q$ to fulfil all rules contained in the data space trust framework related to data sharing. Denote these rules with $R_{\rhd\lhd}$. Data consumer $q$ then will need to provide sufficient attestations $A_{\rhd\lhd} \subseteq A$ to ``convince'' $p$ to trust them.
\item \textbf{Data consumer $q$ perspective.}
For a data transaction to take place, a consumer $q$ will require the data provider $p$ to fulfil all rules contained in the data space trust framework related to data sharing. Denote these rules with $R_{\rhd\rhd}$. Data provider $p$ then will need to provide sufficient attestations $A_{\rhd\rhd} \subseteq A$ to ``convince'' $q$ to trust them.
\end{enumerate}

This allows us to summarize our observation in the following proposition describing all necessary conditions for permitting trusted data transactions in data space $D$..

\begin{proposition}[Necessary general conditions for data sharing]		\label{th-nec-cond-data-sharing}
Let $D = \langle O, T, R, A \rangle$ denote a particular data space with 
$A_\rhd = A_{\rhd\lhd} \cup A_{\rhd\rhd} \subseteq A$. Then the necessary general condition to allow a trusted data transaction to take place
between data provider $p$ and data consumer $q$ are:

\begin{equation}		\label{eq-nec-cond-data-sharing}
	\forall a_i \in A_\rhd: a_1 \land a_2 \land a_3 \land ... \  \Leftrightarrow \  \Diamond( p \rhd q ) \; \forall p,q \in O.
\end{equation}
The operator $\Diamond$ is taken from a suitable modal logic (such as T, S4, or S5 \cite{Hughes-1996}), and, as usually, denotes possibility. Hence, $\Diamond r$ for a proposition $r$ is read as ``It is possible that $r$''.
\end{proposition}

\begin{proof}
See above.
\end{proof}

In the following, we will extend the notation we introduced in proposition (\ref{th-nec-cond-data-sharing}) and write $T_\rhd$ for the set
$T_\rhd = \{ t \in T \ | \  (t,r) \, \vdash \, A_\rhd \}$.

\subsection{Data Space Interoperability}

Despite the ubiquitous requirement of ``data space interoperability'', we are not aware of any true formal conceptualization of this ethereal, as it seems, concept.

\begin{definition}[Data space interoperability]		\label{def-dsp-interop}
Let $U = \langle O^U, T^U, R^U, A^U \rangle$ and $V = \langle O^V, T^V, R^V, A^V \rangle$ be two data spaces.
Then data spaces $U$ and $V$ are \emph{interoperable}, written as $U \bowtie V$, when the following condition holds.
\begin{equation}		\label{eq-dsp-interop}
	U \bowtie V \ \Leftrightarrow \ \Diamond(p^U \rhd q^V) \land \Diamond(r^V \rhd s^U) \quad \forall p^U, s^U \in O^U; \forall q^V, r^V \in O^V.
\end{equation}
\end{definition}

This definition of \emph{interoperability} emphasises the capability of the mutually exchange of data between any two participants of the two data spaces.
The following theorem determines the conditions under which this may be achieved.

\begin{proposition}[Data space interoperability]		 \label{th-dsp-interop}
Let $U = \langle O^U, T^U, R^U, A^U \rangle$ and $V = \langle O^V, T^V, R^V, A^V \rangle$ be two data spaces.
Then
\begin{equation}
	U \bowtie V \  \Leftrightarrow \  T_\rhd^U = T_\rhd^V.
\end{equation}
with $T_\rhd^U$ and $T_\rhd^V$ denoting the subset of trust propositions of the respective data spaces $U$ and $V$ related to the necessary general conditions for data sharing (as per proposition \ref{th-nec-cond-data-sharing}).
\end{proposition}

We need to stress that, in the following, trust propositions $T_\rhd$ are regarded to be extremely generic and to cover, for instance, all the five facets of ISO/IEC DIS 19941-1 \cite{Iso-2026} such as transport, syntactic, semantic, behavioural, and policy interoperability. Additionally, technical standards such as the protocols and participant agents to be used are also assumed to be included in $T_\rhd$.

We will use the following lemma to simplify our proof.

\begin{lemma}			\label{lem-2}
For a one-way data sharing between different data spaces to be possible, the following conditions need to hold for the respective trust propositions.
\begin{equation}
	\Diamond( p^U \rhd q^V) \ \Rightarrow  \  T_{\rhd\rhd}^U \subseteq T_{\rhd\rhd}^V \; \land \; T_{\rhd\lhd}^V \subseteq T_{\rhd\lhd}^U.
\end{equation}
\end{lemma}

\begin{proof}
The data consumer $q^V$ needs assertions $A_{\rhd\rhd}^V$ to hold to accept data from any provider. As $q^V$ only accepts trust propositions
$t \in T_{\rhd\rhd}^V$ and data provider $p^U$ can only supply trust propositions $t \in T_{\rhd\rhd}^U$, this can only be fulfilled in case
$T_{\rhd\rhd}^U \subseteq T_{\rhd\rhd}^V$.

The data provider $p^U$ needs assertions $A_{\rhd\lhd}^U$, but only accepts trust propositions from $T_{\rhd\lhd}^U$.
As the data consumer can only provide trust propositions out of $T_{\rhd\lhd}^V$, this can only be fulfilled in case 
$T_{\rhd\lhd}^V \subseteq T_{\rhd\lhd}^U$.
\end{proof}

Let us now turn to the proof of proposition \ref{th-dsp-interop}.

\begin{proof}
\textbf{Direction $\Rightarrow$} \\
From equ.~(\ref{def-dsp-interop}) we immediately get
\begin{equation}				\label{eq-intermed-1}
	\Diamond(p^U \rhd q^V) \land \Diamond(r^V \rhd s^U) \quad \forall p^U, s^U \in O^U; \forall q^V, r^U \in O^V.
\end{equation}
Applying lemma \ref{lem-2} to the first element of the connective in (\ref{eq-intermed-1}) yields 
\begin{equation*}
	T_{\rhd\rhd}^U \subseteq T_{\rhd\rhd}^V \; \land \; T_{\rhd\lhd}^V \subseteq T_{\rhd\lhd}^U.
\end{equation*}
Applying the lemma to the second element (or part) of the connective yields
\begin{equation*}
	T_{\rhd\rhd}^V \subseteq T_{\rhd\rhd}^U \; \land \; T_{\rhd\lhd}^U \subseteq T_{\rhd\lhd}^V.
\end{equation*}
Combining the left-hand and right hand sides of the preceding two statements, we immediately have $T_{\rhd\rhd}^U = T_{\rhd\rhd}^V$ and $T_{\rhd\lhd}^V \subseteq T_{\rhd\lhd}^U$. Using $T_{\rhd\rhd} \cup T_{\rhd\lhd} = T_\rhd$ for any data space, we get $T_\rhd^U = T_\rhd^V$. \\

\textbf{Direction $\Leftarrow$} \\
Obvious, as both, provider and consumer, accept all trust propositions from the other participant in the data transaction.
\end{proof}

There is promise and pain in the result of proposition \ref{th-dsp-interop}: On one hand, data spaces now have a clear prescription how to achieve interoperability between themselves, namely agreeing on policies, rules, and standards. On the other hand, the proposition shows that the extent of cross-data space interoperability is strictly determined by the amount of policies, rules, and standards mutually shared and accepted.

To be technically precise, it is \textbf{not} directly the set of rules and standards but \emph{just} the set of trust service providers and the evidences or proofs they provide which need to be shared (in the sense of being mutually accepted) between different ecosystems (\emph{arg.}: $T_\rhd^U = T_\rhd^V$ in eqn.~(\ref{th-nec-cond-data-sharing}); the rule sets $R^U$ and $R^V$ are not directly involved).

In particular, it means, that \emph{prior general} agreement is not always needed. This may happen in case a data product also includes verifiable proofs which can be independently checked by the receiving participant. In this case, the strategy changes from ``We trust each other because we share the same policies'' to ``We trust this specific data because it comes with a proof we can verify. In that sense, the real limit of interoperability may not just be the size of the shared policy set, it may be the expressiveness of the shared proof system -- [which is] a materially different boundary condition"
\cite{Riley-2026}. \\

In passing, we note the following pleasing property of our notion of data space interoperability.

\begin{lemma}[Data space interoperability as equivalence relation]
Assume $\mD$ to be a suitable universe of data spaces $U$, $V$, $W \in \mD$ as defined above and denote with $\bowtie \; \subseteq \mD \times \mD$ the data space interoperability relation.

Then $\bowtie$ is an equivalence relation.
\end{lemma}
\begin{proof}
Reflexivity, $U \bowtie U$ is trivial, as well as symmetry, $U \bowtie V \Leftrightarrow V \bowtie U$.
For transitiveness, we easily observe that from $U \bowtie V \land V \bowtie W$ we get $T_\rhd^U = T_\rhd^V = T_\rhd^W$ resulting in $U \bowtie W$.
\end{proof}

This is encapsulating the fact that if all data spaces agree on a set of common trust propositions, they are all equally interoperable to the extent their \emph{individual} rules and policies solely depend on verification via this common ``trust layer''.

\section{Outlook}

We are currently working on extending the ecosystem trust profile with concepts from the latest (25.11) GaiaX Identity, Credential and Access Management Document \cite{Gaia-X-2025c}, most importantly incorporating the much more expanded ontology of a \texttt{Trust Scope Credential}.

\section*{Acknowledgements}

The author greatly acknowledges the overall  conception and higher level architecture of the IMX federated trust use case by Klaus Ottradovetz and the numerous discussions with him and others of the team on turning this vision into reality.

Furthermore, we extend our gratitude to Christoph Lange-Bever, co-lead of the Gaia-X Service Characteristics Working Group who supported the ``coming into life'' of the very first version of the Gaia-X Meta-Registry ontology and to Pietro Bartoccioni, co-lead of the Gaia-X Identity, Credentials, and Access Management (ICAM) Working Group, who provided very convincing and concrete ideas on certain elements of our Ecosystem Trust Profile \emph{avant la lettre}.

This work would not have been possible without the continuous and active support of the whole CTO team and the Gaia-X Lab team in particular, led by Yassir Sellami, who listened to and always gave constructive feedback to all the different embryonic versions of these concepts. Additional kudos to Ryan Reychico who checked the very first instance of the Ecosystem Trust Profile ontology in \texttt{LinkML}. Besides her general technical comments, Delphine Claerhout created the IMX extension in our Gaia-X 3.0 ``Danube'' architecture and an initial version of an underlying IMX ontology. 

Final praise, again, goes to Yassir Sellami for adapting our sophisticated GitLab CI/CD for ontology merge requests and the subsequent fully automated publication pipeline.

% ------------------------------------------------------------------------------------------------
% bibliography
%

\bibliographystyle{ieeetr}

\bibliography{Ecosystems_Literature}

\end{document}